\documentclass{llncs}

\usepackage{amssymb,amsmath,amscd,latexsym,epic,eepic}
\usepackage{times}
\usepackage{gastex}
\usepackage{paralist}
\usepackage{fullpage}

\def\qed{\rule{0.4em}{1.4ex}} 
 \newcommand{\set}[1]{\{#1\}}

\newcommand{\outcome}{\mathrm{Outcome}}
\newcommand{\Prb}{\mathrm{Pr}}

\newcommand{\trans}{\delta}
\newcommand{\dist}{{\cal D}}
\newcommand{\distr}{{\cal D}}

\def\abs#1{\ensuremath{\lvert #1\rvert}}

\newcommand{\un}{\mathsf{un}}
\newcommand{\Cone}{\mathsf{Cone}}
\newcommand{\ObsSeq}{\mathsf{ObsSeq}}

\newcommand{\slopefrac}[2]{\leavevmode\kern.1em
  \raise .5ex\hbox{\the\scriptfont0 #1}\kern-.1em
  /\kern-.15em\lower .25ex\hbox{\the\scriptfont0 #2}}

\newcommand{\straa}{\alpha} \newcommand{\Straa}{{\mathcal A}}
\newcommand{\strab}{\beta} \newcommand{\Strab}{{\mathcal B}}

 \newcommand{\Nats}{\mathbb{N}}

\newcommand{\real}{{\mathbb R}}

\newcommand{\Inf}{\mathrm{Inf}} 

\newcommand{\ov}{\overline}

\newcommand{\Pref}{{\sf Prefs}}
\newcommand{\Play}{{\sf Plays}}
\newcommand{\Last}{{\sf Last}}

\newcommand{\tuple}[1]{\langle #1 \rangle}

\newcommand{\Safe}{\mathsf{Safe}}
\newcommand{\Reach}{\mathsf{Reach}}
\newcommand{\Buchi}{\mathsf{Buchi}}
\newcommand{\coBuchi}{\mathsf{coBuchi}}
\newcommand{\Parity}{\mathsf{Parity}}
\newcommand{\Outcome}{{\mathsf{Outcome}}}
\newcommand{\target}{{\cal T}}
\newcommand{\Obs}{{\cal{O}}}
\newcommand{\obs}{\mathsf{obs}}

\newcommand{\cale}{\mathcal{E}}
\newcommand{\wt}{\widetilde}
\newcommand{\wh}{\widehat}

\makeatletter

\begingroup \catcode `|=0 \catcode `[= 1
\catcode`]=2 \catcode `\{=12 \catcode `\}=12
\catcode`\\=12 |gdef|@xcomment#1\end{comment}[|end[comment]]
|endgroup

\def\@comment{\let\do\@makeother \dospecials\catcode`\^^M=10\def\par{}}

\def\begincomment{\@comment\@xcomment}

\makeatother

\newenvironment{comment}{\begincomment}{}

\title{Equivalence of Games with Probabilistic Uncertainty and Partial-observation Games}
\author{%
Krishnendu Chatterjee\inst{1}%
\and Martin Chmelik\inst{1}
\and Rupak Majumdar\inst{2}}
\institute{%
IST Austria (Institute of Science and Technology Austria)
\and
MPI-SWS, Germany
}
\date{}
\begin{document}
\maketitle
\begin{abstract}
We introduce games with probabilistic uncertainty,
a natural model for controller synthesis in which the
controller observes the state of the system through imprecise sensors 
that provide correct information about the current state with a fixed probability.
That is, in each step, the sensors return an observed state, and 
given the observed state, there is a probability distribution (due to the 
estimation error) over the actual current state.
The controller must base its decision on the observed state (rather than the 
actual current state, which it does not know).
On the other hand, we assume that the environment can perfectly observe
the current state.
We show that our model can be reduced in polynomial time to standard
partial-observation stochastic games, and vice-versa.
As a consequence we establish the precise decidability frontier for the new
class of games, and for most of the decidable problems establish 
optimal complexity results.
\end{abstract}

\section{Introduction}

In a control system, a controller interacts with its
environment through sensors and actuators.
The controller observes the state of the environment 
through a set of sensors,
computes a control signal that depends 
on the history of observed sensor readings,
and feeds the control signal to the environment through actuators. 
The state of the environment is then updated as a function of the control
signal as well as a disturbance signal that models external 
inputs to the environment.
In a \emph{reactive} setting, the sense-compute-actuate cycle repeats
forever, resulting in an infinite trace of environment states.
The objective of the controller is to ensure that the trace belongs to
a given specification of ``good'' traces.
The controller synthesis problem asks, given the dynamical law that
specifies how the environment state changes according to the
controller inputs and external disturbances, 
and a specification of good traces, to synthesize a control
law that ensures that the environment traces are good, no matter how
external disturbances behave.

Controller synthesis has been studied extensively
for deterministic games with $\omega$-regular
specifications~\cite{BuchiLandweber69,Rabin69,KV00b}.
In this setting, the problem is modeled as a game on a graph. 
The vertices of the graph represent system states, 
and are divided into ``controller states''
and ``disturbance states.'' 
At a controller state, the controller
chooses an outgoing edge and moves to a neighboring vertex along this edge.
At a disturbance state, the disturbance chooses an outgoing edge and
moves along this edge.
This continues ad infinitum, defining a sequence of states. 
If this sequence satisfies the specification, the controller wins; otherwise,
the disturbance wins.
The games are called {\em perfect observation}, since both players have
exact knowledge of the current state and the history of the game.

The study of perfect-observation deterministic games have been extended to 
systems with {\em partial observation}, in which the controller can only observe part of the
environment's state~\cite{Reif79,CDHR07}, and to {\em stochastic dynamics}~\cite{FV97,Con92,crg-tcs07,dAM01},
in which the state updates happen according to a probabilistic law.

The ``standard model'' of partial-observation stochastic games~\cite{CDHR07,BD08,BGG09} 
is described as an extension to the above graph model, by fixing an
equivalence relation on the vertices (the ``observation function''),
and stipulating that the controller only sees the equivalence class of
the current vertex, not the particular vertex the state is in.
In addition, the transitions of the graph are stochastic: the
controller and the disturbance each choose some move, and the next
vertex is chosen according to a probability distribution based on the
current vertex and the chosen move.

In this paper, we introduce a different, albeit natural, model of probabilistic uncertainty
in controller synthesis.
Consider a state given by $n$ bits.
The sensors used to measure the state are typically not perfect, and 
observing the state through the sensor results in some bits being flipped with 
some known probability (probabilistic noise).
In applications where the controller observes the state bits through a network,
then the probabilistic noise in the communication channels results in bits
being flipped with some known probability 
(according to the classical Shannon's communication channel model).
Thus, the controller observes $n$ bits through the sensor,
and this estimate defines a probability distribution over the state space for the
current state.
In contrast, we allow the disturbance to precisely observe the state,
corresponding to a worst case assumption on the disturbance.
The objective of the controller is to find a strategy that ensures
that the system satisfies the specification under this probabilistic
uncertainty on the current state.
We distinguish between two models of the disturbance. 
In the first model, the disturbance observes the correct sequence of states as well
as both the observation of the controller and the sequence of controller moves.
In the second model, the disturbance observes the correct sequence of states as well
as the sequence of controller moves (but not the observation of the controller).
It turns out that the two models give rise to subtle differences in defining the probability
measures on the games, as well as different complexities in the solution algorithms.

Our model (which we refer to as games with probabilistic uncertainty) is inspired by analogous 
models of state estimation under probabilistic noise in continuous control systems.
We believe this model of games with probabilistic uncertainty naturally captures
the behavior of many sensor-based control systems.
Intuitively, the standard model of partial-observation games represent ``partial but 
correct information'' where the controller can observe correctly only the first $k<n$ 
bits of the state (i.e., the observation is partial as the controller observes only
a part of the state bits, but the information about the observed state bits is always 
correct). 
In contrast,  our model of games with probabilistic uncertainty 
represent ``complete but uncertain information'' where the controller can observe 
all the $n$ bits of the state but with uncertainty of observation (i.e., the controller
can observe all the bits, but each bit is correct with some probability).
Since the type of uncertain information in our model is very different from the 
standard models of partial-observation games studied
in the literature, the relationship between them is not immediate.

Our main contribution, along with the introduction of the natural model of 
games with probabilistic uncertainty, 
is establishing the equivalence of the 
new class of games and partial-observation games.
Our main technical result is a polynomial-time reduction from this new model 
of games with probabilistic uncertainty to standard partial-observation games, and a converse 
reduction from partially-observable Markov decision processes (POMDPs) to
games with probabilistic uncertainty.
The results to establish the equivalence of the two classes of games which 
represent two different notions of information (partial but correct vs complete
but uncertain) are quite intricate. 
For example, for the new class of games the inductive definition of probability 
measure is subtle and different from the classical definition of 
probability measure for probabilistic systems~\cite{Vardi85,CY95}.
This is because the controller observes a history that can be completely different
from the actual history, whereas the environment (or disturbance) observes the actual history.
We first inductively define a probability measure of observed history,
given the actual history, and use it to define the probability measure inductively. 
We show how our polynomial constructions for reduction capture the subtleties
in the probability measure, and by establishing precise mapping of strategies 
(which is at the heart of the proof of correctness of the reduction) we obtain the 
desired equivalence result.

In the positive direction, our reduction allows us to solve controller
synthesis problems for games with probabilistic uncertainty against
$\omega$-regular specifications, using algorithms of~\cite{CDHR07,BGG09}.
In the negative direction, we get lower bounds on the hardness of
problems by using known lower bounds for POMDPs using the hardness results of~\cite{BBG08,CDH10a}.
In particular, with our reductions we establish precisely the decidability 
frontier of games with probabilistic uncertainty for various classes of 
parity objectives (a canonical form to express $\omega$-regular specifications);
and for most of the decidable problems we establish EXPTIME-complete
bounds, and in some cases 2EXPTIME upper bounds and EXPTIME lower bounds 
(see Table~\ref{tab:complexity}).
Moreover, our reduction allows the rich body of algorithms (such as symbolic
and anti-chain based algorithms~\cite{CDHR07,BGG09}) for partial-observation 
games, along with any future algorithmic developments for partial-observation 
games, to be applicable to solve games with probabilistic uncertainty.
In summary, our results provide precise decidability frontier, optimal complexity
(in most cases), and algorithmic solutions for games with probabilistic uncertainty, 
that is a natural model for control problems with state estimation under 
probabilistic noise. 

\section{Games with Probabilistic Uncertainty}
\label{sec:probimperfect}

In this section we introduce a class of games with probabilistic
imperfect information, and call them games with probabilistic uncertainty.

\smallskip\noindent{\em Probability distribution.}
A \emph{probability distribution} on a finite set $A$ is a function
$\kappa: A \to [0,1]$ such that $\sum_{a \in A} \kappa(a) = 1$. 
We denote by $\dist(A)$ the set of probability distributions on $A$.


\smallskip\noindent{\em Game structures with probabilistic uncertainty.}
A game structure with probabilistic uncertainty consists of a tuple 
$\mathcal{G} = (L, \Sigma_I,\Sigma_O,\Delta,\un)$, where
(a)~$L$ is a set of \emph{locations}; 
(b)~$\Sigma_I$ and $\Sigma_O$ are two sets
of input and output alphabets, respectively;
(c)~$\Delta: L \times \Sigma_I \times \Sigma_O \to \dist(L)$ is a 
probabilistic transition function that given a location, an input 
and an output letter gives the probability distribution over the next 
locations; and 
(d)~$\un: L \to \dist(L)$ is the \emph{probabilistic uncertainty function} 
that given the true current location describes the probability distribution 
of the observed location. 
If $\un$ is the identity function we obtain perfect-observation games.

Intuitively, a game proceeds as follows.
The game starts at some location $\ell\in L$.
Player~1 observes a state drawn from the distribution $\un(\ell)$,
which represents a potentially faulty observation process.
Intuitively, at every step the player can observe the value of all 
variables that corresponds to the state of the game, but there is a probability 
that the observed value of some variables is incorrect.
Player~2 observes the ``correct'' state $\ell$.
Given the observation of the history of the game so far, 
Player~1 picks an input alphabet $\sigma^i\in \Sigma_i$.
Player~2 then picks an output letter $\sigma^o\in\Sigma_o$: we consider two 
variants, (1)~Player~2 only observes the history of correct locations and the moves
of the players; and 
(2)~Player~2 observes the history of correct locations, the moves of the 
players, and also observes the history of observed locations of Player~1.
The state of the game is updated to $\ell'$ with probability $\Delta(\ell,
\sigma^i,\sigma^o)(\ell')$.
This process is repeated ad infinitum.

\smallskip\noindent{\em Plays.}
A \emph{play} of $\mathcal{G}$ is a sequence
$\rho = \ell_0 \sigma_0^i \sigma_0^o \ell_1 \sigma_1^i \sigma_1^o \ldots$
of locations, input letter, and output letter,
such that for all $j\geq 0$ we have $\Delta(\ell_j,\sigma^i_j,
\sigma^o_j)(\ell_{j+1}) > 0$.
The \emph{prefix up to $\ell_n$} of the play~$\rho$ is denoted by $\rho(n)$, 
its \emph{length} is $\abs{\rho(n)} = n+1$ and its \emph{last element} is 
$\Last(\rho(n)) = \ell_n$. The set of plays in $\mathcal{G}$ is denoted by $\Play(\mathcal{G})$,
and the set of corresponding finite prefixes is denoted
$\Pref(\mathcal{G})$.

\smallskip\noindent{\em Strategies.}
A strategy for Player~1 observes the finite prefix of a play and then
selects an input letter (pure strategies) or a probability distribution over
input letters in $\Sigma_i$.
Formally, a pure strategy for Player~1 is a function $\straa :
\Pref(\mathcal{G})\to \Sigma_i$, and a randomized strategy for Player~1 is
a function $\straa: \Pref(\mathcal{G})\to \dist(\Sigma_i)$.
Similarly, pure and randomized strategies for Player~2 are defined as functions
$\strab : \Pref(\mathcal{G}) \times \Sigma_i \to \Sigma_o$ and
$\strab : \Pref(\mathcal{G}) \times \Sigma_i \to \dist(\Sigma_o)$,
respectively.
Note that Player~2 sees Player~1's choice of input action at each step.
In the case where Player~2 observes also the history of observed locations, the pure 
and randomized strategies are defined as functions
$\strab : \Pref(\mathcal{G}) \times \Pref(\mathcal{G}) \times \Sigma_i \to \Sigma_o$ and
$\strab : \Pref(\mathcal{G})  \times \Pref(\mathcal{G}) \times \Sigma_i \to \dist(\Sigma_o)$,
respectively, where the output letter is chosen based on the original history and 
observed history.
We refer to strategies that observes both histories as ``all-powerful'' strategies 
for Player~2.

\smallskip\noindent{\em Outcomes.}
The \emph{outcome} of two randomized strategies $\straa$ for Player~1
and $\strab$ for Player~2 from a location $\ell\in L$ is the set of plays
$\rho = \ell_0 \sigma^i_0\sigma^o_0\ldots$ such that
(1) $\ell = \ell_0$, 
(2) there exists a sequence $\ell_0' \ell_1'\ldots$ such that
$\un(\ell_j)(\ell_j') > 0$ for each $j \geq 0$,
(3) for each $j\geq 0$, we have
$\straa(\ell_0'\sigma^i_0\sigma^o_0\ldots \ell_j')(\sigma^i_j) > 0$
and $\strab(\rho(j), \sigma^i_j)(\sigma^o_j) > 0$ 
(if $\strab$ is an all-powerful strategy, then 
$\strab(\rho(j), \ell_0' \sigma^i_0 \sigma^o_0 \ell_1' \ldots \ell_j', \sigma^i_j)(\sigma^o_j) > 0$), 
and $\Delta(\ell_j, \sigma^i_j,\sigma^o_j)(\ell_{j+1}) > 0$.
The primed sequence $\ell_0'\ell_1'\ldots$ gives the sequence of
observations made by Player~1 using the probabilistic uncertainty function.
Note that this sequence may be incorrect with some 
probability due to probabilistic uncertainty in the observation.
We denote this set of plays as $\outcome(\mathcal{G},\ell,
\straa,\strab)$.
The outcome of two pure strategies is defined analogously, considering
pure strategies as degenerate randomized strategies which pick a
letter with probability one.
The \emph{outcome set} of the pure (resp.\ randomized)
strategy $\straa$ for Player~$1$ in $\mathcal{G}$ is the set
$\Outcome_1(\mathcal{G},\ell,\straa)$ of plays $\rho$ such that there exists a
pure (resp.\ randomized) strategy $\strab$ for
Player~$2$ with $\rho\in\outcome(\mathcal{G},\ell,\straa,\strab)$.
The outcome set $\Outcome_2(\mathcal{G},\ell,\strab)$ for Player~2 is defined symmetrically.


\newcommand{\ActMt}{\mathsf{ActMt}}

\smallskip\noindent{\em Probability measure.} Given strategies $\straa$ and 
$\strab$, we define the probability measure $\Prb_{\ell_0}^{\straa,\strab}(\cdot)$.
The definition of the probability measure is subtle and non-standard as the 
prefix that Player~1 observes can be completely different from the original 
history.
For a finite prefix $\rho\in \Pref(\mathcal{G})$, let $\Cone(\rho)$ denote
the set of plays with $\rho$ as prefix.
We will define $\Prb_{\ell_0}^{\straa,\strab}(\cdot)$ for cones, and then by 
Caratheodory extension theorem~\cite{Billingsley} there is a unique extension to all 
measurable sets of paths.
To define the probability measure we also need to define 
a function $\ObsSeq(\rho)$, that given a finite prefix $\rho$, gives the probability 
distribution over finite prefixes $\rho'$, such that $\ObsSeq(\rho)(\rho')$ denotes the 
probability of observing $\rho'$ given the correct prefix is $\rho$.
The base case is as follows:
\[
\Prb_{\ell_0}^{\straa,\strab}(\Cone(\ell_0))=1; \qquad
\ObsSeq(\ell_0)(\ell')=\un(\ell_0)(\ell').
\]
The inductive definition of $\ObsSeq$ is as follows: for a prefix $\rho$ of length $n+1$
\[
\ObsSeq(\rho \sigma^i_n \sigma^o_n \ell_{n+1})(\rho' \sigma^i_n \sigma^o_n \ell_{n+1}')=
\ObsSeq(\rho)(\rho')\cdot \un(\ell_{n+1})(\ell_{n+1}')
\]
Given a sequence $\rho=\ell_0 \sigma^i_0 \sigma^o_0 \ell_1 \sigma^i_1 \sigma^o_1 \ldots \ell_n$,
we define $\ActMt(\rho)=\set{\wt{\rho}=  
\wt{\ell}_0 \wt{\sigma}^i_0 \wt{\sigma}^o_0 \ell_1 \wt{\sigma}^i_1 \wt{\sigma}^o_1 \ldots \wt{\ell}_n
\mid \forall 1 \leq j \leq n-1. \ \wt{\sigma}^i_j= \sigma^i_j \text{ and } 
 \wt{\sigma}^o_j =\sigma^o_j }$ 
the sequences of same length as $\rho$ such that the sequence of input and 
output letter matches (i.e., the set of action-matching prefixes).
Note that for non action-matching prefixes the observation sequence function 
always assigns probability zero. 
The inductive case for the probability measure is as follows: for a prefix $\rho$ of length $n+1$ with last state $\ell_n$, 
we have 
\[
\begin{array}{l}
\Prb_{\ell_0}^{\straa,\strab}(\Cone(\rho \sigma^i_n\sigma^o_n \ell_{n+1}))
= \\
\displaystyle
\ \ \Prb_{\ell_0}^{\straa,\strab}(\Cone(\rho)) \cdot \bigg( \sum_{\rho'\in \ActMt(\rho)} 
\ObsSeq(\rho)(\rho') \cdot \straa(\rho')(\sigma^i_n) \cdot \strab(\rho \sigma^i_n)(\sigma^o_n) \cdot 
\Delta(\ell_n, \sigma^i_n,\sigma^o_n)(\ell_{n+1})\bigg);
\end{array}
\]
i.e., $\ObsSeq(\rho)(\rho')$ gives the probability to observe 
$\rho'$, then $\straa(\rho')(\sigma^i_n)$ denotes the probability to play 
$\sigma_n^i$ given the strategy and observed sequence $\rho'$, and since 
Player~2 observes the correct sequence the probability to play $\sigma_n^o$
is given by $\strab(\rho \sigma_n^i)(\sigma_n^o)$ (Player~2 observes $\rho$),
and the final term $\Delta(\ell_n,\sigma_n^i,\sigma_n^o)(\ell_{n+1})$ gives
the transition probability.
If $\strab$ is an all-powerful strategy, then $\strab$ observes both the correct history 
$\rho$ and the observed history $\rho'$, and then the definition is as follows:
\[
\begin{array}{l}
\Prb_{\ell_0}^{\straa,\strab}(\Cone(\rho \sigma^i_n\sigma^0_n \ell_{n+1}))
= \\
\displaystyle
\ \ \Prb_{\ell_0}^{\straa,\strab}(\Cone(\rho)) \cdot \bigg( \sum_{\rho'\in \ActMt(\rho)} 
\ObsSeq(\rho)(\rho') \cdot \straa(\rho')(\sigma^i_n) \cdot \strab(\rho,\rho',\sigma^i_n)(\sigma^o_n) \cdot 
\Delta(\ell_n, \sigma^i_n,\sigma^o_n)(\ell_{n+1})\bigg).
\end{array}
\]

\smallskip\noindent{\em Winning objectives.}
An \emph{objective} for Player~$1$ in $\mathcal{G}$ is a set $\phi \subseteq \Play(\mathcal{G})$
of plays.
A play $\rho \in \Play(\mathcal{G})$ 
\emph{satisfies} the objective $\phi$, denoted $\rho \models \phi$, if $\rho \in \phi$.
We consider $\omega$-regular objectives 
specified as parity objectives (a canonical form to express all $\omega$-regular 
objectives~\cite{Thomas97}).
For a play $\rho = \ell_0 \sigma^i_0\sigma^o_0\ldots$,
we denote by $\rho_k$ the $k$-th 
location $\ell_k$ of the play and denote by $\Inf(\rho)$ 
the set of locations that 
occur infinitely often in $\rho$, that is, 
$\Inf(\rho)=\{ \ell \mid \forall i \exists j: j> i \text{ and } \ell_j=\ell\}$.
We consider the following classes of objectives.

\begin{enumerate}
\item \emph{Reachability and safety objectives.}
Given a set $\target \subseteq L$ of target locations, the \emph{reachability} objective 
$\Reach(\target)$ requires that a location in $\target$ be visited at least once, 
that is, $\Reach(\target)=\set{ \rho  \mid \exists k \geq 0 \cdot \rho_k \in \target}$. 
Dually, the \emph{safety} objective $\Safe(\target)$ requires that only states 
in $\target$ be visited.
Formally,
$\Safe(\target)=\set{ \rho \mid \forall k \geq 0 \cdot \rho_k \in \target}$.

\item \emph{B\"uchi and coB\"uchi objectives.}
Let $\target\subseteq L$ be a set of target locations.
The \emph{B\"uchi} objective $\Buchi(\target)$ 
requires that a state in $\target$ be visited infinitely often, 
that is,
$\Buchi(\target)=\set{ \rho \mid \Inf(\rho) \cap \target \neq \emptyset}$.
Dually, the \emph{coB\"uchi} objective $\coBuchi(\target)$ requires that 
only states in $\target$ be visited infinitely often.
Formally, $\coBuchi(\target) =\set{\rho \mid \Inf(\rho) \subseteq \target}$.

\item \emph{Parity objectives.} 
For $d \in \Nats$, let $p:L \to \{0,1,\ldots,d\}$ be a 
\emph{priority function}, 
which maps each state to a nonnegative integer priority.
The \emph{parity} objective $\Parity(p)$ requires that the minimum priority 
that occurs infinitely often be even.
Formally, $\Parity(p)=\set{\rho \mid \min\set{ p(\ell) \mid \ell \in \Inf(\rho)} 
\mbox{ is even} }$.
The B\"uchi and coB\"uchi objectives are the special cases 
of parity objectives with two priorities, $p: L \to \set{0,1}$ and 
$p: L \to \set{1,2}$, respectively.


\end{enumerate}

\smallskip\noindent{\em Sure, almost-sure and positive winning.}
An \emph{event} is a measurable set of plays, and 
given strategies $\straa$ and $\strab$ for the two players, 
the probabilities of events are uniquely defined.
For an objective~$\phi$, assumed to be Borel, we denote by 
$\Prb_{\ell}^{\straa,\strab}(\phi)$ 
the probability that $\phi$ is satisfied by the play obtained from the 
starting location  $\ell$ when the 
strategies $\straa$ and $\strab$ are used.
Given a game $\mathcal{G}$, an objective $\phi$, and a location $\ell$, 
we consider the following winning modes: 
(1)~a strategy $\straa$ for Player~1 is \emph{sure winning} for the objective 
$\phi$ from $\ell\in L$ if $\outcome(\mathcal{G},\ell,\straa,\strab) \subseteq 
\phi$ for all strategies $\strab$ for Player~$2$;
(2)~a strategy $\straa$ for Player~$1$ is 
\emph{almost-sure winning} for the objective $\phi$ from $\ell\in L$ if 
$\Prb_{\ell}^{\straa,\strab}(\phi)=1$ for all strategies $\strab$ for 
Player~$2$; and 
(3)~a strategy $\straa$ for Player~$1$ is 
\emph{positive winning} for the objective $\phi$ from $\ell\in L$ if 
$\Prb_{\ell}^{\straa,\strab}(\phi)>0$ for all strategies $\strab$ for 
Player~$2$.

Qualitative analysis of a game consists of the computation of the sure, almost-sure and
positive winning sets.
The sure (resp. almost-sure and positive) winning decision problem for an objective
consists of a game and a starting location $\ell$, and asks whether
there is a
sure (resp. almost-sure and positive) winning strategy 
from $\ell$.

\section{Partial-observation Stochastic Games}\label{sec:def}

We now recall the usual definition of partial-observation games and their 
subclasses.
We focus on partial-observation turn-based probabilistic games,
where at each round one of the players is in charge of choosing the next 
action and the transition function is probabilistic.
We will present a polynomial time reduction of games with probabilistic uncertainty 
to these games.

\smallskip\noindent{\bf Partial-observation games.}
A \emph{partial-observation stochastic game} 
(for short partial-observation game or simply a \emph{game}) is a tuple 
$G=\tuple{S_1 \cup S_2,A_1,A_2,\trans_1 \cup \trans_2,\Obs_1,\Obs_2}$ with the following components: 
\begin{enumerate}
\item \emph{(State space).} $S=S_1 \cup S_2$ is a finite set of states, where $S_1 \cap S_2=\emptyset$ (i.e.,
$S_1$ and $S_2$ are disjoint), states in $S_1$ are Player~1 states, and states in $S_2$ are Player~2 states. 
\item \emph{(Actions).} $A_i$ ($i=1,2$) is a finite set of actions for Player~$i$. 
\item \emph{(Transition function).} For $i\in \set{1,2}$, the probabilistic 
transition function for Player~$i$ is the function $\trans_i : S_i \times A_i \to \dist(S_{3-i})$ 
that maps a state $s_i \in S_i$ and an action~$a_i \in A_i$ to the probability 
distribution  $\trans_i(s_i,a_i)$ over the successor states in $S_{3-i}$ (i.e., games are alternating). 
\item ~\emph{(Observations).} $\Obs_1 \subseteq 2^{S}$ is a finite set 
of observations for Player~$1$ that partitions the state space~$S$, and similarly
$\Obs_2$ is the observations for Player~2.
These partitions uniquely define functions $\obs_i: S \to \Obs_i$, for $i \in \set{1,2}$, 
that map each state to its observation such that $s \in \obs_i(s)$ for all $s \in S$.
We will also consider the special case of one-sided games, where Player~2 is perfectly 
informed (has complete observation), i.e., 
$\Obs_2=S$, and $\obs_2(s)=s$ for all $s\in S$ (i.e., the partition consists
of singleton states).
\end{enumerate}

\smallskip\noindent{\bf Special Class: POMDPs.}
We will consider one special class of partial-observation games called 
\emph{partial-observable Markov decision processes} (POMDPs), 
where the action set for Player~2 is a singleton (i.e., there is effectively only 
Player~1 and stochastic transitions). 
Hence we will omit the action set and observation for Player~2 and represent 
a POMDP as the following tuple $G=\tuple{S, A, \trans, \Obs}$, where 
$\trans: S \times A \to \distr(S)$.

\smallskip\noindent{\em Plays.}
In a game, in each turn, for $i \in \set{1,2}$, if the current state $s$ is in $S_i$, 
then Player~$i$ chooses an action $a \in A_i$, and the successor state is 
chosen by sampling the probability distribution $\trans_i(s,a)$.
A \emph{play} in $G$ is an infinite sequence of states and actions 
$\rho=s_0 a_0 s_1 a_1 \ldots$ such that for all $j \geq 0$, if $s_j \in S_i$, 
for $i\in \set{1,2}$, then $a_j \in A_i$ such that 
$\trans_i(s_j,a_j)(s_{j+1})>0$.
The definitions of prefix and length are analogous
to the definitions in Section~\ref{sec:probimperfect}.
For $i \in \set{1,2}$, we denote by $\Pref_i(G)$ the set of finite prefixes 
in $G$ that end in a state in $S_i$.
The \emph{observation sequence} of $\rho= s_0 a_0 s_1 a_1\ldots$ for Player~$i$ ($i=1,2$) 
is the unique infinite sequence of observations and actions, i.e., 
$\obs(\rho) = o_0 a_0 o_1 a_1 o_2 \ldots$ such that $s_j \in o_j$ for all $j \geq 0$.
The observation sequence for finite sequences (prefix of plays) is defined 
analogously.

\smallskip\noindent{\em Strategies.}
A \emph{pure strategy} in $G$ for Player~$1$ is a function $\straa:\Pref_1(G) \to A_1$. 
A \emph{randomized strategy} in $G$ for Player~$1$ is a function $\straa:\Pref_1(G) \to \dist(A_1)$. 
A (pure or randomized) strategy $\straa$ for Player~$1$ is 
\emph{observation-based} if for all prefixes $\rho,\rho' \in \Pref_1(G)$, 
if $\obs(\rho)=\obs(\rho')$, then $\straa(\rho)=\straa(\rho')$. 
We omit analogous definitions of strategies for Player~$2$.
We denote by $\Straa_G$, $\Straa_G^O$, $\Straa_G^P$, $\Strab_G$, $\Strab_G^O$, $\Strab_G^P$ 
the set of all Player-$1$ strategies in~$G$, the set of all observation-based Player-$1$ strategies, 
the set of all pure Player-$1$ strategies, 
the set of all Player-$2$ strategies in~$G$, the set of all observation-based Player-$2$ strategies, 
and the set of all pure Player-$2$ strategies, 
respectively.
In the setting where Player~$1$ has partial-observation and Player~$2$ 
has complete observation, the set $\Strab_G$ of all strategies coincides with 
the set $\Strab_G^O$ of all observation-based strategies.
We will require the players to play observation-based strategies.

\smallskip\noindent{\em Outcomes.}
The \emph{outcome} of two randomized strategies $\straa$ (for
Player~$1$) and $\strab$ (for Player~$2$) from a state $s$ in $G$ is the set of 
plays $\rho=s_0 a_0 s_1 a_1 \ldots \in \Play(G)$, with $s_0=s$, where for all $j \geq 0$, 
if $s_j \in S_1$ (resp. $s_j \in S_2$), then 
$\straa(\rho(j))(a_j) > 0$ (resp. $\strab(\rho(j))(a_j)>0$) and 
$\trans_1(s_j,a_j)(s_{j+1})>0$ (resp. $\trans_2(s_j,a_j)(s_{j+1})>0$).
This set is denoted $\outcome(G,s,\straa,\strab)$.
The outcome of two pure strategies is defined analogously by viewing  
pure strategies as randomized strategies that play their chosen action
with probability one.
The \emph{outcome set} of the pure (resp.\ randomized)
strategy $\straa$ for Player~$1$ in $G$ is the set
$\Outcome_1(G,s,\straa)$ of plays $\rho$ such that there exists a
pure (resp.\ randomized) strategy $\strab$ for
Player~$2$ with $\rho\in\outcome(G,s,\straa,\strab)$.
The outcome set $\Outcome_2(G,s,\strab)$ for Player~2 is defined symmetrically.

\smallskip\noindent{\em Probability measure.} 
We define the probability measure $\Prb_{s}^{\straa,\strab}(\cdot)$ as follows:
for a finite prefix $\rho$, let $\Cone(\rho)$ denote the set of plays
with $\rho$ as prefix.
Then we have $\Prb_{s}^{\straa,\strab}(\Cone(s))=1$, and for a prefix 
of length $n$ ending in a Player~1 state $s_n$ we have
\[
\Prb_{s}^{\straa,\strab}(\Cone(\rho a_n s_{n+1})) = 
\Prb_{s}^{\straa,\strab}(\Cone(\rho)) \cdot \straa(\rho)(a_n) \cdot \trans_1(s_n,a_n)(s_{n+1});
\] 
and the definition when $s_n$ is a Player~2 state is similar.
For a set $Q$ of finite prefixes, we write 
$\Prb_{s}^{\straa,\strab}(\Cone(Q))$ for 
$\Prb_{s}^{\straa,\strab}(\bigcup_{\rho\in Q}\Cone(\rho))$.

The winning modes sure, almost-sure, and positive are defined
analogously to Section~\ref{sec:probimperfect}, where we restrict the players
to play an observation-based strategy.
From the results of~\cite{CDHR07,BGG09,BBG08,BD08,CDH10a} we obtain  the following 
theorem summarizing the results for partial-observation games and POMDPs.

\begin{theorem}[\cite{CDHR07,BGG09,BBG08,BD08,CDH10a}]\label{thrm1}
The following assertions hold: 
\begin{enumerate}
\item \emph{(One-sided games and POMDPs).} 
The sure, almost-sure and positive winning for safety objectives;
the sure and almost-sure winning for reachability objectives and B\"uchi
objectives; the sure and positive winning for coB\"uchi objectives; and
the sure winning for parity objectives are 
EXPTIME-complete for one-sided partial-observation games (Player~2 
perfectly informed) and POMDPs.
The positive winning problem for reachability objectives is PTIME-complete 
both for one-sided partial-observation games and POMDPs.

\item \emph{(General partial-observation games).}
The sure, almost-sure winning for safety objectives, 
the sure winning for parity objectives are EXPTIME-complete for partial-observation 
games; 
the almost-sure winning for reachability objectives and B\"uchi objectives, and
the positive winning for safety and coB\"uchi objectives are 2EXPTIME-complete
for partial-observation games.
The positive winning problem for reachability objectives is EXPTIME-complete.

\item \emph{(Undecidability results).}
The positive winning problem for B\"uchi objectives, the almost-sure winning
problem for coB\"uchi objectives, and the positive and almost-sure winning problems 
for parity objectives are undecidable for POMDPs.

\end{enumerate}
\end{theorem}

\section{Reduction: Games with Probabilistic Uncertainty to Partial-observation Games}

We now present a reduction of games with probabilistic uncertainty to 
classical partial-observation games.
Let $G=(L,\Sigma_I,\Sigma_O,\Delta,\un)$ be a game with probabilistic 
uncertainty and we construct a partial-observation game $H= 
(L \times L \cup L \times L \times \Sigma_I, A_1=\Sigma_I, A_2= \Sigma_O, 
\trans=\trans_1 \cup \trans_2, \Obs_1,\Obs_2)$ as follows (below as $\trans_1$
and $\trans_2$ would be clear from context, we simply use $\trans$ for
simplicity):
\begin{enumerate}
\item The transition function $\trans_1$ is deterministic and 
for $(\ell_1,\ell_2) \in L\times L$ and $\sigma_I \in \Sigma_I$ 
we have 
\[
\trans((\ell_1,\ell_2),\sigma_I)=(\ell_1,\ell_2,\sigma_I)
\]

\item The transition function $\trans_2$ captures both 
$\Delta$ and $\un$ and is defined as follows: for $(\ell_1,\ell_2,\sigma_I) \in 
L\times L \times \Sigma_I$ and $\sigma_O \in \Sigma_O$ 
we have 
\[
\trans((\ell_1,\ell_2,\sigma_I),\sigma_O)(\ell_1',\ell_2')
=\Delta(\ell_1,\sigma_I,\sigma_O)(\ell_1') \cdot \un(\ell_1')(\ell_2').
\]
Intuitively, the first component of the game $H$ keeps track of the 
real state of the game $G$, and the second component 
keeps track of the information available from probabilistic 
uncertainty.
Hence Player~1 is only allowed to observe the second component 
which is the probability distribution over the observable state 
given the current state.

\item The observation mapping is as follows: we have $\Obs_1=L$; and 
$\obs_1(\ell_1,\ell_2)=\obs_1(\ell_1,\ell_2,\sigma_I)=\ell_2$,
i.e., only the second component is observable.
We will consider two cases for $\Obs_2$: for the reduction of all-powerful 
strategies we will consider Player~2 has complete-observation, 
and in the other case we have $\Obs_2=L$ and Player~2 observes the first component
that represents the correct history: i.e.,
$\obs_2(\ell_1,\ell_2)=\obs_2(\ell_1,\ell_2,\sigma_I)=\ell_1$.

\item For a parity objective in $G$ given by priority function 
$p_G: L \to \set{0,1,\ldots, d}$, we consider the priority function
$p_H$ in $H$ as follows:
$p_H((\ell,\ell'))=p_H((\ell,\ell',\sigma_I))=p_G(\ell)$, for all 
$\ell,\ell'\in L$ and $\sigma_I \in \Sigma_I$.

\end{enumerate}

\smallskip\noindent{\bf Correspondence of strategies.} 
We will now establish the correspondence of 
probabilistic uncertain strategies in $G$ and the observation
based strategies in $H$.
We present a few notations.
For simplicity of presentation, we will use a slight abuse of notation:
given a history (or finite prefix) 
$\rho_H=s_0 a_0 s_1 a_1 s_2 a_2\ldots s_{2n}$ in $H$ we will represent 
the history as 
$s_0 a_0 a_1 s_2 a_2 a_3 s_3 \ldots s_{2n}$ as the intermediate state is always 
uniquely defined by the state and the action.
Intuitively this is removing the stuttering and does not affect parity objectives.

\smallskip\noindent{\em Mapping of strategies from $G$ to $H$.} 
Given a history $\rho_H= s_0 a_0 a_1 s_2 a_2 a_3 s_3 \ldots s_{2n}$ in $H$, 
such that $s_{2i}=(\ell_{2i}^1,\ell_{2i}^2)$, we consider two histories 
in $G$ as follows: 
\[
g_1(\rho_H)= \ell_0^1 a_0 a_1 \ell_2^1 a_2 a_3 \ldots \ell_{2n}^1; \qquad
g_2(\rho_H)= \ell_0^2 a_0 a_1 \ell_2^2 a_2 a_3 \ldots \ell_{2n}^2.
\]
Intuitively, $g_1$ gives the first component (which is the correct history) and $g_2$ gives
the second component (which is the observed history).
We now define the mapping of strategies from $G$ to $H$: 
given strategy $\straa_G$ for Player~1, a strategy $\strab_G$ for Player~2, 
and  an all-powerful strategy $\strab_G^A$ for Player~2, in the 
game $G$, we define the corresponding strategies in $H$ as follows: 
for a history $\rho_H$ and an action $a_i$ for Player~1
we have 
\[
\begin{array}{rcl}
\straa_H(\rho_H) & = & \straa_G(g_2(\rho_H)); \\[2ex]
\strab_H(\rho_H \ a_i) & = & \strab_G(g_1(\rho_H) \ a_i); \\[2ex]
\strab_H^C(\rho_H\ a_i) & = & \strab_G^A(g_1(\rho_H), g_2(\rho_H),a_i). 
\end{array}
\]
Note that $\straa_H$ and $\strab_H$ are observation-based strategies, and 
$\strab_H^C$ is a strategy with complete-observation,
i.e., all-powerful strategies are mapped to complete-observation strategies.
Hence for all-powerful strategies the reduction is to one-sided games.
We will use $\wh{g}$ to denote the mapping of strategies, i.e.,
$\straa_H=\wh{g}(\straa_G)$, $\strab_H=\wh{g}(\strab_G)$, and
$\strab_H^C=\wh{g}(\strab_G^A)$.

\smallskip\noindent{\em Mapping of strategies from $H$ to $G$.} We now present 
the mapping in the other direction. 
Let 
$\rho_G^1=\ell_0^1 \sigma_0^i \sigma_0^o \ell_1^1 \sigma_1^i \sigma_1^o \ldots 
\ell_n^1$, and
$\rho_G^2=\ell_0^2 \sigma_0^i \sigma_0^o \ell_1^2 \sigma_1^i \sigma_1^o \ldots 
\ell_n^2$ be two prefixes in $G$.
Intuitively, the first represent the correct history and the second the 
observed history.
Then we consider the following set of histories in $H$:
\[
h_1(\rho_G^1)= \set{\rho_H \mid g_1(\rho_H)=\rho_G^1}; \qquad
h_2(\rho_G^2)= \set{\rho_H \mid g_2(\rho_H)=\rho_G^2}; 
\]
and 
\[
h_{12}(\rho_G^1,\rho_G^2)= 
(\ell_0^1,\ell_0^2) \sigma_0^i \sigma_0^o (\ell_1^1,\ell_1^2) 
\sigma_1^i \sigma_1^o \ldots (\ell_n^1,\ell_n^2).
\]
We now define the mapping of strategies.
Given an observation-based strategy $\straa_H \in \Straa_H^O$ for Player~1, 
observation-based strategy $\strab_H\in \Strab_H^O$ for Player~2, and 
complete observation-based strategy $\strab_H^C \in \Strab_H$, we define the 
following strategies in $G$: for a correct history $\rho_G^1$, 
observed history $\rho_G^2$, and input $\sigma^i$ 
we have
\[
\begin{array}{rcl}
\strab_G(\rho_G^1 \ \sigma^i) & = & 
\strab_H(\rho_H \ \sigma^i); \quad \rho_H\in h_1(\rho_G^1); \\[2ex]
\straa_G(\rho_G^2) & = & \straa_H(\rho_H); \qquad  \rho_H\in h_2(\rho_G^2); \\[2ex]
\strab_G^A(\rho_G^1,\rho_G^2,\sigma^i) & = &
\strab_H^C(h_{12}(\rho_G^1,\rho_G^2),\sigma_i).
\end{array}
\] 
Note that since $\strab_H$ is observation-based it plays the same for 
all $\rho_H\in h_1(\rho_G^1)$, and similarly,
since $\straa_H$ is observation-based it plays the same for 
all $\rho_H\in h_2(\rho_G^2)$.
Also observe that the strategy $\strab_G^A$ is an all-powerful strategy.
We will use $\wh{h}$ to denote the mapping of strategies, i.e.,
$\straa_G=\wh{h}(\straa_H)$, $\strab_G=\wh{h}(\strab_H)$, and
$\strab_G^A=\wh{h}(\strab_H^C)$. 

Given a starting state $\ell_0 \in G$, consider the following 
probability distribution $\mu$ in $H$: $\mu(\ell_0,\ell)=\un(\ell_0)(\ell)$.
Given the mapping of strategies, our goal is to establish the equivalences of
the probability measure. 
We introduce some notations required to establish the equivalence. 
For $j\geq 0$, we denote by $(\tau_j^1,\tau_j^2)$ the pair of random variables
to denote the $j$-th Player~1 state of the game $H$, and by 
$\theta_j^i$ and $\theta_j^o$ the random variables for the actions following
the $j$-th state. 
Our first lemma establishes a connection of the probability of observing 
the second component in $H$ given the first component along with function 
$\ObsSeq$.
We introduce notations to define two events: given two prefixes 
$\rho_G^1=\ell_0^1 \sigma_0^i \sigma_0^o \ell_1^1 \sigma_1^i \sigma_1^o \ldots 
\ell_n^1$, and
$\rho_G^2=\ell_0^2 \sigma_0^i \sigma_0^o \ell_1^2 \sigma_1^i \sigma_1^o \ldots 
\ell_n^2$ in $G$, 
let $\cale_{1,2}(\rho_G^1,\rho_G^2)$ denote the event that for all 
$0 \leq j \leq n$ we have $\tau_j^1=\ell_j^1, \tau_j^2 =\ell_j^2$ and for all 
$0 \leq j \leq n-1$ we have $\theta_j^i= \sigma_j^i, \theta_j^o=\sigma_j^o$; 
and $\cale_1(\rho_G^1)$ denote the event that  
for all $0 \leq j \leq n$ we have $\tau_j^1=\ell_j^1$ and for all $0 \leq j \leq n-1$ 
we have $\theta_j^i= \sigma_j^i, \theta_j^o=\sigma_j^o$.

\begin{lemma}\label{lemm-proof1}
Let $\rho_G^1=\ell_0^1 \sigma_0^i \sigma_0^o \ell_1^1 \sigma_1^i \sigma_1^o \ldots 
\ell_n^1$, and
$\rho_G^2=\ell_0^2 \sigma_0^i \sigma_0^o \ell_1^2 \sigma_1^i \sigma_1^o \ldots 
\ell_n^2$ be two prefixes in $G$.
Then for all strategies $\straa_H$ and $\strab_H$, the probability that the second component 
sequence in $H$ is $\rho_G^2$, given the first component sequence is $\rho_G^1$ is 
$\ObsSeq(\rho_G^1)(\rho_G^2)$, i.e., formally 
\[
\Prb^{\straa_H,\strab_H}_{\mu}(\cale_{1,2}(\rho_G^1,\rho_G^2) \mid \cale_1(\rho_G^1))
=\ObsSeq(\rho_G^1)(\rho_G^2). 
\]
\end{lemma}
\begin{proof}
The proof is by induction on the length of the prefixes. 
The base case is as follows: let the length of prefixes $\rho_G^1$ and $\rho_G^2$ be 1, 
with $\rho_G^1 = \ell_0$ and $\rho_G^2=\ell$.
Then we have
\[
\ObsSeq(\ell_0)(\ell) = \mu(\ell_{0},\ell);
\]
as required.
We now consider the inductive case: we consider prefixes 
$\rho_G^1 \sigma_n^i \sigma_n^o \ell_{n+1}^1$ and 
$\rho_G^2 \sigma_n^i \sigma_n^o \ell_{n+1}^2$. 
Let us consider the events 
$\cale_{n+1}^1=\cale_{1,2}(\rho_G^1 \sigma_n^i \sigma_n^o \ell_{n+1}^1,
\rho_G^2 \sigma_n^i \sigma_n^o \ell_{n+1}^2) $ 
and
$\cale_{n+1}^2=\cale_{1}(\rho_G^1 \sigma_n^i \sigma_n^o \ell_{n+1}^1)$. 
Let $\ov{\cale}_{n+1}^1$ denote the event that 
$\tau_n^1=\ell_n^1$, $\tau_n^2=\ell_n^2$, 
$\tau_{n+1}^1=\ell_{n+1}^1, \tau_{n+1}^2=\ell_{n+1}^2, 
\theta_{n}^1=\sigma_n^i$, and $\theta_{n}^2=\sigma_{n}^o$; 
and $\ov{\cale}_{n+1}^2$ denote the event that
$\tau_n^1=\ell_n^1$, $\tau_{n+1}^1=\ell_{n+1}^1,\theta_{n}^1=\sigma_n^i$, 
and $\theta_{n}^2=\sigma_{n}^o$.
Then by definition we have 
\[
\begin{array}{rcl}
\Prb_{\mu}^{\straa_H,\strab_H}(\ov{\cale}_{n+1}^1 \mid \ov{\cale}_{n+1}^2) 
& = & 
\displaystyle
\frac{\delta((\ell_n^1,\ell_n^2, \sigma_n^i),\sigma_n^o)(\ell_{n+1}^1,\ell_{n+1}^2)}
{\sum_{\wt{\ell}_{n}^2, \wt{\ell}_{n+1}^2} 
\delta((\ell_n^1,\wt{\ell}_n^2, \sigma_n^i),\sigma_n^o)(\ell_{n+1}^1,\wt{\ell}_{n+1}^2)} \\[3.5ex] 
 & & \quad \qquad \text{(In the numerator all choices are fixed, and } \\[2ex] 
 & & \quad \qquad \text{in denominator are all possible choices 
of the second component)} \\[3.5ex]
& = & 
\displaystyle
\frac{
\Delta(\ell_n^1, \sigma_n^i,\sigma_n^o)(\ell_{n+1}^1) \cdot \un(\ell_{n+1}^1)(\ell_{n+1}^2)
}
{
\Delta(\ell_n^1, \sigma_n^i,\sigma_n^o)(\ell_{n+1}^1) \cdot 
\sum_{\wt{\ell}_{n+1}^2} 
\un(\ell_{n+1}^1)(\wt{\ell}_{n+1}^2)
} \\[3.5ex]
& = & 
\un(\ell_{n+1}^1)(\ell_{n+1}^2)  
\qquad \quad \text{(Since $\sum_{\wt{\ell}_{n+1}^2} \un(\ell_{n+1}^1)(\wt{\ell}_{n+1}^2)=1$)}
\end{array}
\] 
Note that the crucial fact used in the above proof is in the second equality and 
the fact is that for all 
$\wt{\ell}_n^2$ we have 
$\delta((\ell_n^1,\wt{\ell}_n^2, \sigma_n^i),\sigma_n^o)(\ell_{n+1}^1,\wt{\ell}_{n+1}^2)
=\Delta(\ell_n^1, \sigma_n^i,\sigma_n^o)(\ell_{n+1}) \cdot \un(\ell_{n+1})(\wt{\ell}_{n+1}^2)$
(i.e., it is independent of $\wt{\ell}_n^2$).
Hence using the above equality and inductive hypothesis we have: 
\[
\begin{array}{rcl}
\Prb^{\straa_H,\strab_H}_{\mu}(\cale_{n+1}^1 \mid \cale_{n+1}^2) 
& = & \displaystyle 
\Prb^{\straa_H,\strab_H}_{\mu}(\cale_{1,2}(\rho_G^1,\rho_G^2) \mid \cale_{1}(\rho_G^1)) \cdot 
\Prb^{\straa_H,\strab_H}_{\mu}(\ov{\cale}_{n+1}^1 \mid \ov{\cale}_{n+1}^2) \\[2ex]
& = & 
\ObsSeq(\rho_G^1)(\rho_G^2) \cdot 
\Prb^{\straa_H,\strab_H}_{\mu}(\ov{\cale}_{n+1}^1 \mid \ov{\cale}_{n+1}^2) \qquad\text{(By inductive hypothesis)} \\[2ex]
& = & 
\ObsSeq(\rho_G^1)(\rho_G^2) \cdot 
\un(\ell_{n+1}^1)(\ell_{n+1}^2) \qquad \text{(By previous equality)} \\[2ex] 
& = & \ObsSeq(\rho_G^1 \sigma_n^i \sigma_n^o \ell_{n+1}^1)(\rho_G^2 \sigma_n^i \sigma_n^o \ell_{n+1}^2) 
\end{array}
\] 
The desired result follows. 
\hfill\qed
\end{proof}

We will now establish the equivalences of the probabilities of the cones.

\begin{lemma}\label{lemm-proof2}
For all finite prefixes $\rho_G^1$ in $G$, the following assertions hold:
\begin{enumerate}
\item For all strategies $\straa_G$, $\strab_G$, $\strab_G^A$ (all-powerful), 
we have
\[
\Prb^{\straa_G,\strab_G}_{\ell_0}(\Cone(\rho_G^1)) =
\Prb^{\wh{g}(\straa_G),\wh{g}(\strab_G)}_{\mu}(\Cone(h_1(\rho_G^1))); 
\qquad
\Prb^{\straa_G,\strab_G^A}_{\ell_0}(\Cone(\rho_G^1)) =
\Prb^{\wh{g}(\straa_G),\wh{g}(\strab_G^A)}_{\mu}(\Cone(h_1(\rho_G^1))). 
\]
\item For all strategies $\straa_H$, $\strab_H$, $\strab_H^C$ (complete-observation), 
we have
\[
\Prb^{\wh{h}(\straa_H),\wh{h}(\strab_H)}_{\ell_0}(\Cone(\rho_G^1)) =
\Prb^{\straa_H,\strab_H}_{\mu}(\Cone(h_1(\rho_G^1))); 
\qquad
\Prb^{\wh{h}(\straa_H),\wh{h}(\strab_H^C)}_{\ell_0}(\Cone(\rho_G^1)) =
\Prb^{\straa_H,\strab_H^C}_{\mu}(\Cone(h_1(\rho_G^1))). 
\]
\end{enumerate}
\end{lemma}
\begin{proof}
We will present the result for the first item, and the proof for second 
item is identical.
Let us denote by $\straa_H=\wh{g}(\straa_G)$ and $\strab_H=\wh{g}(\strab_G)$.
We will prove the result by induction on the length of the prefixes.
The base case is as follows:
let the length of the prefix $\rho_G^1$ be 1, with $\rho_G^1 = \ell_0$.
We observe that 
$\Prb^{\straa_G,\strab_G}_{\ell_0}(\Cone(\ell_0)) = 1$, and 
$\Prb^{\straa_H,\strab_H}_{\mu}(\Cone(h_1(\ell_{0}))) = 1$, and 
for all other cones of length $1$ the probability is zero. 
This completes the base case.

We now consider the inductive case: by inductive hypothesis we assume that 
$\Prb^{\straa_G,\strab_G}_{\ell_0}(\Cone(\rho_G^1)) = \Prb^{\straa_H,\strab_H}_{\mu}(\Cone(h_1(\rho_G^1 )))$; and 
show that
\[
\Prb^{\straa_G,\strab_G}_{\ell_0}(\Cone(\rho_G^1 a_{n} b_{n} \ell_{n+1})) = 
\Prb^{\straa_H,\strab_H}_{\mu}(\Cone(h_1(\rho_G^1 a_{n} b_{n} \ell_{n+1}))).
\]
Let $\ell_n$ be the last state of $\rho_G^1$.
We first consider the left-hand side (LHS):
\[
\begin{array}{rcl}
\Prb&^{\straa_G,\strab_G}_{\ell_0}&(\Cone(\rho_G^1 a_{n} b_{n} \ell_{n+1})) 
\\ & = &\displaystyle  
\Prb_{\ell_0}^{\straa_{G},\strab_{G}}(\Cone(\rho^{1}_{G})) \cdot \bigg( \sum_{\rho' \in \ActMt(\rho_G^1)} 
\ObsSeq(\rho^{1}_{G})(\rho') 
\cdot \straa_{G}(\rho')(a_{n}) \cdot \strab_{G}(\rho^{1}_{G} a_{n} )(b_{n}) \cdot 
\Delta(\ell_{n}, a_{n}, b_{n})(\ell_{n+1})\bigg) \\[3ex]
& = & 
\displaystyle\Prb^{\straa_H,\strab_H}_{\mu}(\Cone(h_1(\rho_G^1))) \cdot \bigg( \sum_{\rho'\in \ActMt(\rho_G^1)} 
\ObsSeq(\rho^{1}_{G})(\rho') 
\cdot \straa_{G}(\rho')(a_{n}) \cdot \strab_{G}(\rho^{1}_{G} a_{n} )(b_{n}) \cdot \Delta(\ell_{n}, a_{n}, b_{n})(\ell_{n+1})\bigg) \\[3ex]
& = & 
\displaystyle \sum_{\rho'\in \ActMt(\rho_G^1)} \Prb^{\straa_H,\strab_H}_{\mu}(\Cone(h_{12}(\rho_G^1,\rho')))
 \cdot \straa_{G}(\rho')(a_{n}) \cdot \strab_{G}(\rho^{1}_{G} a_{n} )(b_{n}) \cdot \Delta(\ell_{n}, a_{n}, b_{n})(\ell_{n+1})
\end{array}
\]
Above the first equality is by definition, the second equality by inductive hypothesis,
and the last equality is obtained from Lemma~\ref{lemm-proof1} as follows:
by Lemma~\ref{lemm-proof1} we have 
$\ObsSeq(\rho^{1}_{G})(\rho')  =
\Prb^{\straa_H,\strab_H}_{\mu}(\cale_{1,2}(\rho_G^1,\rho') \mid \cale_1(\rho_G^1))$,
and hence
\[
\begin{array}{rcl}
\Prb^{\straa_H,\strab_H}_{\mu}(\Cone(h_1(\rho_G^1)))  \cdot  
& \displaystyle \sum_{\rho'\in \ActMt(\rho_G^1)} & \ObsSeq(\rho^{1}_{G})(\rho') \\[3ex]
& = &   
\displaystyle
\sum_{\rho'\in \ActMt(\rho_G^1)} \Prb^{\straa_H,\strab_H}_{\mu}(\Cone(h_1(\rho_G^1))) \cdot  
\Prb^{\straa_H,\strab_H}_{\mu}(\cale_{1,2}(\rho_G^1,\rho') \mid \cale_1(\rho_G^1)) \\[3ex]
& = & 
\displaystyle
\sum_{\rho'\in \ActMt(\rho_G^1)} \Prb^{\straa_H,\strab_H}_{\mu}(\Cone(h_{12}(\rho_G^1,\rho'))).
\end{array} 
\]

We now consider the right-hand side (RHS) $\Prb^{\straa_H,\strab_H}_{\mu}(\Cone(h_1(\rho_G^1 a_{n} b_{n} \ell_{n+1})))$ 
and the RHS can be expanded as: (below for brevity we write $\wh{\rho}=h_{12}(\rho_{G}^1,\rho')$)
\[
\sum_{\rho'\in \ActMt(\rho_G^1)} \sum_{\ell_{n+1}'} 
\Prb^{\straa_H,\strab_H}_{\mu} (\Cone(\wh{\rho})) \cdot \straa_{H}(\wh{\rho})(a_{n}) \cdot  
\strab_{H}(\wh{\rho} a_{n})(b_{n}) \cdot \delta((\ell_{n},\ell'_{n},a_{n}),b_{n})(\ell_{n+1},\ell'_{n+1}) 
\]
Since we have
\[
\straa_{H}(h_{12}(\rho_G^1, \rho'))(a_{n}) = \straa_{G}(\rho')(a_{n}); 
\qquad \text{and} \qquad
\strab_{H}(h_{12}(\rho_{G}^1, \rho') a_{n})(b_{n}) =  \strab_{G}(\rho^{1}_{G} a_{n} )(b_{n}), 
\]
the above expression for RHS is equivalently described as:
\[
\sum_{\rho'\in \ActMt(\rho_G^1)} \sum_{\ell_{n+1}'} 
\Prb^{\straa_H,\strab_H}_{\mu}(\Cone(h_{12}(\rho_G^1,\rho')))
 \cdot \straa_{G}(\rho')(a_{n}) \cdot \strab_{G}(\rho^{1}_{G} a_{n} )(b_{n}) \cdot \Delta(\ell_{n}, a_{n}, b_{n})(\ell_{n+1}) \cdot \un(\ell_{n+1})(\ell'_{n+1}) 
\]
Since $\sum_{\ell_{n+1}'}  \un(\ell_{n+1})(\ell'_{n+1}) =1$, it follows that LHS is equal to the RHS.
The result for correspondence for all-powerful strategy $\strab_G^A$ is essentially copy-paste of
the above proof replacing appropriately $\strab_G$ by $\strab_G^A$.
This completes the proof and the desired result follows.
\hfill\qed
\end{proof}

It follows that there is a sure,  almost-sure, positive 
winning strategy in $G$ for $\Parity(p_G)$ iff there is a corresponding one in 
$H$ for $\Parity(p_H)$ and hence from Theorem~\ref{thrm1} we obtain the 
following result.

\begin{theorem}\label{lemm1}
The following assertions hold:
\begin{enumerate}

\item \emph{(All-powerful Player~2).}
The sure, almost-sure and positive winning for safety objectives;
the sure and almost-sure winning for reachability objectives and B\"uchi
objectives; the sure and positive winning for coB\"uchi objectives; and
the sure winning for parity objectives can be solved in 
EXPTIME for games with probabilistic uncertainty with all-powerful strategies
for Player~2.
The positive winning for reachability objectives can be solved in PTIME.

\item \emph{(Not all-powerful Player~2).}
The sure, almost-sure winning for safety objectives; and 
the sure winning for parity objectives can be solved in EXPTIME;
the almost-sure winning for reachability objectives and B\"uchi
objectives; the positive winning for safety and coB\"uchi objectives
can be solved in 2EXPTIME for games with probabilistic uncertainty 
without all-powerful strategies for Player~2.
The positive winning for reachability objectives can be solved in EXPTIME.
\end{enumerate}
\end{theorem}

\section{Reduction: POMDPs to Games with Probabilistic Uncertainty}
In this section we present a reduction in the reverse direction and show that
POMDPs with parity objectives can be reduced to games with probabilistic 
uncertainty and parity objectives.
We first present the reduction and then show the correctness of the reduction
by mapping prefixes, strategies, and establishing the equivalence of the 
probability measure.

\smallskip\noindent{\bf Reduction: POMDPs to games with probabilistic uncertainty.}
Let $H = (S, A, \trans, \Obs)$ be a POMDP with a parity objective $\phi$, 
we construct the game of probabilistic uncertainty $G = (L, \Sigma_{I}, \Sigma_{O}, \Delta, \un)$ as follows:
\begin{itemize}
\item $L = S$; 
\item $\Sigma_{I} = A$;
\item $\Sigma_{O} =\set{ \bot }$;
\item For $\ell \in L$ and $a \in \Sigma_{I}$ let $\Delta (\ell,a,\bot)(\ell')=\trans(\ell,a)(\ell')$, 
i.e., the transition function is same as the transition function of the POMDP.
In other words, the state space is the same, the action choices of the POMDP 
corresponds to the input action choice, and the output action set is singleton,
and the transition function mimics the transition function of the POMDP.
Below we use the probabilistic uncertainty to capture the partial-observation of
the POMDP. 

\item The uncertainty function is as follows:
$\un(\ell)(\ell') = \left\{ 
  \begin{array}{l l}
    0 & \qquad \text{if $\obs(\ell) \not = \obs(\ell')$}\\
    \frac{1}{\vert \obs(\ell) \vert} & \qquad \text{if $\obs(\ell) = \obs(\ell')$}\\
  \end{array} \right.
$
\end{itemize}
The parity objective is the same as the original parity objective.

\smallskip\noindent{\bf Mapping of prefixes.}
Given a prefix (or a finite history) $\rho_H= s_0 a_0 s_1 a_1 s_2 \ldots s_{n}$ in $H$ we construct a 
prefix in $G$ as $\rho_G= s_0 a_0 \bot s_1 a_1 \bot s_2 \ldots s_{n}$ by simply
inserting the $\bot$ actions. This construction defines a bijection 
$h: \Pref_{H} \rightarrow \Pref_{G}$ between prefixes. We can naturally extend the mapping to sets of prefixes. 
Let $\Psi \subseteq \Pref_{H}$, then $h'(\Psi) = \set{ h(\rho) \mid \rho \in \Psi }.$

\begin{lemma}
\label{lem:prefix}
For prefixes $\rho, \rho'$ in $G$ the following assertion holds:

$
\ObsSeq(\rho)(\rho') = \left\{
\begin{array}{l l}
\displaystyle 
\frac{1}{\prod_{i=1}^{n} \vert o_{i}\vert } & \qquad \text{If } \obs(h^{-1}(\rho)) = \obs(h^{-1}(\rho')) = o_{1} a_{1} o_{2} \ldots a_{n-1} o_{n} \\[2ex]
0 & \qquad \text{Otherwise} \\
\end{array}
\right.
$
\end{lemma}
\begin{proof}
We prove the result by induction on the length of prefixes. 
We will only consider $\rho$ and $\rho'$ that have the same length, as 
otherwise by definition the observation sequence probability is~0.
We first consider the base case.

\smallskip\noindent{\em Base case.}
Let $\ell_{0}$ be the initial state. Then $\rho = \ell_{0}$ and 
let $\rho' = \ell$ for some $\ell \in L$. 
Then:
\[
\ObsSeq(\ell_{0}, \ell) = \un(\ell_{0},\ell) = \frac{1}{\vert \obs(\ell_0) \vert}
\] 
if $\ell_{0}$ and $\ell$ have the same observation and $0$  otherwise.
This proves the base case.

\smallskip\noindent{\em Inductive step.}
We now consider prefixes of length $n+1$, and by inductive hypothesis
the result holds for prefixes of length $n$. 
Then 
\[
\ObsSeq(\rho a_{n} \bot \ell_{n+1})(\rho' a_{n} \bot \ell'_{n+1}) = 
\ObsSeq(\rho)(\rho') \cdot \un(\ell_{n+1})(\ell'_{n+1}).
\]
We now consider two cases to complete the proof.
\begin{itemize}
\item If  
$\obs(h^{-1}(\rho a_{n} \bot \ell_{n+1})) \not =  
\obs(h^{-1}(\rho' a_{n} \bot \ell'_{n+1}))$, 
then either $\obs(h^{-1}(\rho)) \not  = \obs(h^{-1}(\rho'))$ 
or $\obs(\ell_{n+1}) \not = \obs(\ell'_{n+1})$. 
It follows that one of the factors ($\ObsSeq(\rho)(\rho')$ or 
$\un(\ell_{n+1})(\ell'_{n+1}$)) is equal to $0$ and hence:
\[
\ObsSeq(\rho a_{n} \bot \ell_{n+1})(\rho' a_{n} \bot \ell'_{n+1}) = 0 
\]

\item Otherwise, we have $\obs(h^{-1}(\rho a_{n} \bot \ell_{n+1})) =  
\obs(h^{-1}(\rho' a_{n} \bot \ell'_{n+1})) = 
o_{1} a_{1} o_{2} \ldots a_{n-1} o_{n} a_{n} o_{n+1}$.
Then:
\[
\ObsSeq(\rho a_{n} \bot \ell_{n+1})(\rho' a_{n} \bot \ell'_{n+1}) = 
\ObsSeq(\rho)(\rho') \cdot \un(\ell_{n+1})(\ell'_{n+1}) = 
\frac{1}{\prod_{i=1}^{n} \vert o_{i}\vert } \cdot \frac{1}{\vert o_{n+1} \vert} = \frac{1}{\prod_{i=1}^{n+1} \vert o_{i}\vert }
\]
\end{itemize}
The desired result follows.
\hfill\qed
\end{proof}

\smallskip\noindent{\bf Mapping of strategies.} We first present 
the mapping of strategies from $H$ to $G$ and then from $G$ to $H$.
Note that in the game $G$, there is no choice for Player~2, and hence
we remove the Player~2 strategies in the descriptions below.

\smallskip\noindent{\em Mapping strategies from $H$ to $G$.}
Let $\alpha_{H}$ be an observation-based Player-1 strategy in $H$ and 
$\rho_{G} =s_0 a_0 \bot s_1 a_1 \bot s_2 \ldots s_{n} $ be a prefix in $G$. 
We define a Player-1 strategy $\alpha_{G}$ in $G$ as follows:
$\alpha_{G} (\rho_{G}) = \alpha_{H} (h^{-1}(\rho_{G}))$.

\smallskip\noindent{\em Mapping strategies from $G$ to $H$.}
Let $\straa_{G}$ be a Player-1 strategy in $G$ and 
$\rho_{H} = s_0 a_0 s_1 a_1  s_2 \ldots s_{n}$ be a prefix in $H$ with 
$o = o_{0} a_{0} o_{1} a_{1} o_{2} \ldots o_{n}$ as its observation sequence. 
Note that as Player~2 has only one strategy (always playing $\bot$) 
we omit it from discussion.
Note that every $\rho \in \ActMt(h(\rho_{H}))$ can have different actions with different probabilities enabled. We define a Player-1 strategy $\straa_{H}$ in $H$ as follows: for an action $a \in A$ we have
\[
\straa_{H}(\rho_{H})(a) =  
\displaystyle 
\sum_{\rho' \in \ActMt(h(\rho_{H}))} \ObsSeq(h(\rho_{H}))(\rho')\, \cdot \, \straa_{G}(\rho')(a). 
\]
We now show that the strategy $\straa_H$ is an observation-based 
strategy for Player~1 in the POMDP.

\begin{lemma}
The strategy $\straa_H$ obtained from strategy $\straa_G$ is an observation-based 
strategy for Player~1 in $H$.
\end{lemma}
\begin{proof}
Let $\rho_H$ and $\rho_H'$ be two prefixes in $H$ that match in 
observation sequence and we need to argue that $\straa_H$ plays the same 
for both prefixes $\rho_H$ and $\rho_H'$.
Observe that since $\rho_H$ and $\rho_H'$ has the same observation sequence,
we have $\ActMt(h(\rho_H))=\ActMt(h(\rho_H'))$.
Moreover it follows from Lemma~\ref{lem:prefix} that 
$\ObsSeq(h(\rho_H))$ only depends on the observation sequence of 
$\rho_H$ and hence for all $\rho' \in 
\ActMt(h(\rho_H))= \ActMt(h(\rho_H'))$ we have 
$\ObsSeq(h(\rho_{H}))(\rho')=\ObsSeq(h(\rho_{H}'))(\rho')$.
It follows that for all actions $a \in A$ we have $\straa_H(\rho_H)(a)=
\straa_H(\rho_H')(a)$.
It follows that $\straa_H$ is observation based.
\hfill\qed
\end{proof}

\smallskip\noindent{\bf Correspondence of probabilities.}
In the following two lemmas we establish the correspondence of the 
probabilities for the mappings.

\begin{lemma}
\label{lem:HtoG}
Let us consider the mapping of strategies from $ H $ to $ G$.
For all prefixes $\rho_{H}$ in $H$ we have 
\[
\Prb^{\straa_H}_{\mu}(\Cone(\rho_H)) = \Prb^{\straa_{G}}_{\ell_{0}}(\Cone(h(\rho_{H}))).
\]
\end{lemma}
\begin{proof}
The proof is based on induction 
on the length of the prefix $\rho_{H}$. We denote the last state of $\rho_{H}$ by $\ell_{n}$.

\smallskip\noindent{\em Base case.}
For prefixes of length 1 where $\rho_{H} = \ell_{0}$ we get 
$\Prb^{\straa_H}_{\mu}(\Cone(\ell_{0})) = 1$ and $\Prb^{\straa_{G}}_{l_{0}}(\Cone(h(\ell_{0}))) = 1$. 
For all other prefixes both sides are equal to $0$.
Hence the base case follows.

\smallskip\noindent{\em Inductive step.}
By inductive hypothesis we assume the result for prefixes $\rho_H$ of length $n$
(i.e., we assume that  $\Prb^{\straa_H}_{\mu}(\Cone(\rho_H)) 
= \Prb^{\straa_G}_{\ell_0}(\Cone(h(\rho_H)))$) and will show that 
\[
\Prb^{\straa_H}_{\mu}(\Cone(\rho_H a_{n} \ell_{n+1})) = \Prb^{\straa_G}_{\ell_0}(\Cone(h(\rho_H a_n \ell_{n+1}))).
\]
First we expand the left hand side (LHS) and by definition we get that:
\[
\Prb^{\straa_H}_{\mu}(\Cone(\rho_H a_{n} \ell_{n+1})) = 
\Prb^{\straa_H}_{\mu}(\Cone(\rho_H)) \cdot \straa_{H}(\rho_H)(a_{n}) \cdot \delta(\ell_{n},a_{n})(\ell_{n+1}). 
\]
We now expand the right hand side (RHS) and get that:
\[
\begin{array}{l}
\Prb^{\straa_{G}}_{\ell_{0}}(\Cone(h(\rho_{H} a_{n} \ell_{n+1}))) = \\
\displaystyle \qquad 
\Prb^{\straa_{G}}_{\ell_{0}}(\Cone(h(\rho_{H} ))) \cdot \left( \sum_{\rho' \in \ActMt(h(\rho_H))} \ObsSeq(h(\rho_{H}))(\rho') 
\cdot \straa_{G}(\rho')(a_n) \cdot \Delta(\ell_{n},a_{n}, \bot)(\ell_{n+1}) \right) 
\end{array}
\]
Using the inductive hypothesis, the definition of the game, and the mapping of strategies we get on the RHS:
\[
\begin{array}{l}
\Prb^{\straa_{G}}_{\ell_{0}}(\Cone(h(\rho_{H} a_{n} \ell_{n+1}))) = \\
\displaystyle \qquad 
\Prb^{\straa_H}_{\mu}(\Cone(\rho_H)) \cdot \left( \sum_{\rho' \in \ActMt(h(\rho_H))} \ObsSeq(h(\rho_{H}))(\rho') 
\cdot \straa_{H}(h^{-1}(\rho'))(a_n) \cdot \delta(\ell_{n},a_{n})(\ell_{n+1}) \right) 
\end{array}
\]
For all $\rho'$ that do not match the observation sequence of $h(\rho_H)$, we have 
$\ObsSeq(h(\rho_H))(\rho')=0$ (by Lemma~\ref{lem:prefix}),
and as $\straa_{H}$ is observation based for all $\rho' \in \ActMt(\rho_{H})$ that matches
the observation sequence of $h(\rho_H)$, the strategy $\straa_H$ plays the same.
Let us denote by $\rho' \approx h(\rho_H)$ that $\rho'$ matches the observation sequence 
of $h(\rho_H)$. 
Then we have 
\[
\begin{array}{rcl}
\displaystyle
\sum_{\rho' \in \ActMt(h(\rho_H))} & \ObsSeq(h(\rho_{H}))(\rho') 
\cdot & \straa_{H}(h^{-1}(\rho'))(a_n)  \\ 
 & = & 
\displaystyle
\sum_{\rho' \in \ActMt(h(\rho_H)), \rho' \approx h(\rho_H)} \ObsSeq(h(\rho_{H}))(\rho') 
\cdot \straa_{H}(h^{-1}(\rho'))(a_n) \\[4ex]
& = &
\displaystyle
\sum_{\rho' \in \ActMt(h(\rho_H)), \rho' \approx h(\rho_H)} \ObsSeq(h(\rho_{H}))(\rho') 
\cdot \straa_{H}(\rho_H)(a_n) \\[4ex]
& = & \straa_{H}(\rho_H)(a_n); 
\end{array}
\]
where the first equality follows as for all sequences $\rho'$ that do 
not match the observation sequence of $h(\rho_H)$ we have 
$\ObsSeq(h(\rho_H))(\rho')=0$; 
the second equality follows as for all $\rho' \approx h(\rho_H)$ we have 
$\straa_H(h^{-1}(\rho'))(a_n)=\straa_H(\rho_H)(a_n)$ (as $\straa_H$ is observation based);
and the last equality follows because as $\ObsSeq$ is a probability distribution we have 
$\sum_{\rho' \in \ActMt(h(\rho_H)), \rho' \approx h(\rho_H)} 
\ObsSeq(h(\rho_{H}))(\rho') = 1$.
Hence we have 
\[
\Prb^{\straa_{G}}_{\ell_{0}}(\Cone(h(\rho_{H} a_{n} \ell_{n+1}))) = 
\Prb^{\straa_H}_{\mu}(\Cone(\rho_H))  \cdot \straa_{H}(\rho_{H})(a_n)  \cdot \delta(\ell_{n},a_{n})(\ell_{n+1}) 
\]
Thus we have 
that LHS and RHS coincide and this completes the proof.
\hfill\qed
\end{proof}

\begin{lemma}
Let us consider the mapping of strategies from $ G $ to $ H$.
For all prefixes $\rho_{G}$ in $G$ we have
\[
\Prb^{\straa_H}_{\mu}(\Cone(h^{-1}(\rho_G))) = \Prb^{\straa_{G}}_{\ell_{0}}(\Cone(\rho_{G})) 
\]
\end{lemma}
\begin{proof}
The inductive proof is as follows and we will denote the last state of $\rho_G$ as $\ell_n$.
The base case is similar to the base case of Lemma~\ref{lem:HtoG}.
We now present the inductive case.

\smallskip\noindent{\em Inductive step.}
By inductive hypothesis we assume the result for prefixes $\rho_G$ of length $n$
(i.e., we assume that  
$\Prb^{\straa_H}_{\mu}(\Cone(h^{-1}(\rho_G))) = \Prb^{\straa_{G}}_{\ell_{0}}(\Cone(\rho_{G}))$) 
and will show that
\[
\Prb^{\straa_H}_{\mu}(\Cone(h^{-1}(\rho_G a_n \ell_{n+1}))) = \Prb^{\straa_{G}}_{\ell_{0}}(\Cone(\rho_{G} a_n \ell_{n+1})).
\] 
First we expand the right hand side (RHS) and by definition we get that:
\[
\Prb^{\straa_{G}}_{\ell_{0}}(\Cone(\rho_{G} a_{n} \ell_{n+1})) = 
\Prb^{\straa_{G}}_{\ell_{0}}(\Cone(\rho_{G} )) 
\cdot \left( \sum_{\rho' \in \ActMt(\rho_G)} 
\ObsSeq(\rho_{G})(\rho') \cdot 
\straa_{G}(\rho')(a_n) \cdot \Delta(\ell_{n},a_{n}, \bot)(\ell_{n+1}) \right) 
\]
As $\Delta(\ell_{n},a_{n}, \bot)(\ell_{n+1})$ does not depend on $\rho'$ we get:
\[
\Prb^{\straa_{G}}_{\ell_{0}}(\Cone(\rho_{G} a_{n} \ell_{n+1})) = 
\Prb^{\straa_{G}}_{\ell_{0}}(\Cone(\rho_{G} )) 
\cdot \Delta(\ell_{n},a_{n}, \bot)(\ell_{n+1}) \cdot 
\left( \sum_{\rho' \in \ActMt(\rho_G)} \ObsSeq(\rho_{G})(\rho') 
\cdot \straa_{G}(\rho')(a_n)  \right) 
\]
We will now show that the expansion of the left hand side (LHS) also gives
the same expression.
Let $\rho_H=h^{-1}(\rho_G)$.
By expanding the LHS we get:
\[
\begin{array}{rcl}
\Prb^{\straa_H}_{\mu}(\Cone(h^{-1}(\rho_G a_{n} \ell_{n+1}))) 
& = & 
\Prb^{\straa_H}_{\mu}(\Cone(h^{-1}(\rho_G))) \cdot \straa_{H}(h^{-1}(\rho_G))(a_{n}) \cdot \delta(\ell_{n},a_{n})(\ell_{n+1}) \\[2ex]
& = & 
\Prb^{\straa_H}_{\mu}(\Cone(\rho_H)) \cdot \straa_{H}(\rho_H)(a_{n}) \cdot \delta(\ell_{n},a_{n})(\ell_{n+1}) \\[2ex]
& = & 
\Prb^{\straa_H}_{\mu}(\Cone(\rho_H)) \cdot \straa_{H}(\rho_H)(a_{n}) \cdot \Delta(\ell_{n},a_{n},\bot)(\ell_{n+1}) \\[2ex]
& = & 
\Prb^{\straa_G}_{\ell_0}(\Cone(\rho_G)) \cdot \straa_{H}(\rho_H)(a_{n}) \cdot \Delta(\ell_{n},a_{n},\bot)(\ell_{n+1}); 
\end{array}
\]
where the first equality is by definition; the second equality is by simply re-writing $h^{-1}(\rho_G)$ as 
$\rho_H$; the third equality is by the definition of $\Delta$ and $\delta$; and the final 
equality is the inductive hypothesis.
By definition of $\straa_H$ we have 
$\straa_H(\rho_H)(a_n)= \left( \sum_{\rho' \in \ActMt(\rho_G)} \ObsSeq(\rho_{G})(\rho') 
\cdot \straa_{G}(\rho')(a_n)  \right)$; and hence it follows 
that LHS and RHS coincide.
Thus the desired result follows.
\hfill\qed
\end{proof}

\begin{comment}
\smallskip\noindent{\bf Reverse reduction for lower bounds from blind POMDPs.}
We now present a reduction from blind POMDPs to games with probabilistic 
uncertainty.
The basic intuition is as follows: even though input player observes the 
location, all locations are equally likely.
Thus even with the observation of the location there is no information 
gained, i.e., we faithfully mimic the blind POMDP.
Formally given a blind POMDP
$H=(S_1 \cup S_2, A_1, A_2=\set{\bot}, \trans_1 \cup \trans_2)$ we construct 
the game with probabilistic uncertainty 
$G=(S_1 \cup S_2, \Sigma_I=A_1,\Sigma_O=A_2=\set{\bot},\Delta)$, where 
$\Delta$ is defined as follows: for $s \in S_1$ and $\sigma_I \in \Sigma_I$
we have $\Delta(s,\sigma_I,\bot)(s')=\trans_1(s,a_1)(s')$; 
and for $s \in S_2$ and $\sigma_I \in \Sigma_I$ we have  
$\Delta(s,\sigma_I,\bot)(s')=\trans(s,\bot)(s')$.
The probabilistic uncertainty function is as follows: 
for $L=S$ and $\ell,\ell' \in L$ we have 
$\un(\ell)(\ell')=\frac{1}{|L|}$.
The correspondence between strategies in $G$ and $H$ is straightforward (similar to 
the case presented above) and
\end{comment}
The previous two lemmas establish the equivalence of the probability measure
and completes the reduction of POMDPs to games with probabilistic uncertainty. 
Hence the lower bounds for POMDPs also gives us the lower bound for 
games with probabilistic uncertainty.
Hence Theorem~\ref{lemm1}, along with the reduction from POMDPs and 
Theorem~\ref{thrm1} gives us the following result for games with 
probabilistic uncertainty (the results are also summarized in 
Table~\ref{tab:complexity}).

\begin{theorem}\label{thrm2}
The following assertions hold: 
\begin{enumerate}
\item \emph{(All-powerful Player~2).}
The sure, almost-sure and positive winning for safety objectives;
the sure and almost-sure winning for reachability objectives and B\"uchi
objectives; the sure and positive winning for coB\"uchi objectives; and
the sure winning for parity objectives are all 
EXPTIME-complete 
for games with probabilistic uncertainty with
all-powerful strategies for Player~2.
The positive winning for reachability objectives is PTIME-complete.

\item \emph{(Not all-powerful Player~2).}
The sure, almost-sure winning for safety objectives; and 
the sure winning for parity objectives are all EXPTIME-complete;
the almost-sure winning for reachability objectives and B\"uchi
objectives; the positive winning for safety and coB\"uchi objectives
can be solved in 2EXPTIME and is EXPTIME-hard for games with probabilistic 
uncertainty without all-powerful strategies for Player~2.
The positive winning for reachability objectives can be solved in EXPTIME.

\item \emph{(Undecidability results).} 
The positive winning problem for B\"uchi objectives, the almost-sure winning
problem for coB\"uchi objectives, and the positive and almost-sure winning problem 
for parity objectives are undecidable for games with probabilistic uncertainty.

\end{enumerate}
\end{theorem}

\begin{table}[h]
\begin{center}
\begin{scriptsize}
\begin{tabular}{|l|c|c|c|c|c|c|}
\cline{2-7}

\multicolumn{1}{l}{}              & \multicolumn{2}{|c|}{Sure}          & \multicolumn{2}{|c|}{Almost}          & \multicolumn{2}{|c|}{Positive} \\
\cline{2-7}
\multicolumn{1}{l|}{{\small \strut}} & All-powerful     & Not-all-powerful               &   All-powerful     & Not-all-powerful    &  All-powerful     & Not-all-powerful \\
\hline
Safety {\small \strut}           & EXP-complete   &  EXP-complete               &   EXP-complete   &  EXP-complete              &  EXP-complete   &  2EXP, EXP     \\
\hline
Reachability {\small \strut}           & EXP-complete   &  EXP-complete               &   EXP-complete   &  2EXP, EXP              &  PTIME-complete   &  EXP, PTIME     \\
\hline
B\"uchi {\small \strut}           & EXP-complete   &  EXP-complete               &   EXP-complete   &  2EXP, EXP              &  Undec.   &  Undec.     \\
\hline
coB\"uchi {\small \strut}           & EXP-complete   &  EXP-complete               &  Undec.   &  Undec.     &   EXP-complete   &  2EXP, EXP              \\
\hline
Parity {\small \strut}           & EXP-complete   &  EXP-complete               &  Undec.   &  Undec.     &   Undec.   &  Undec.              \\
\hline
\hline
\end{tabular}
\end{scriptsize}
\end{center}
\caption{Complexity of games with probabilistic uncertainty with parity objectives, where for each entry we present 
the upper and lower bound, or undecidability.}\label{tab:complexity}
\end{table}

\section{Conclusion}
In this work we considered games with probabilistic uncertainty, which is 
natural for many problems, and has not been considered before. 
We present a reduction of such games to classical partial-observation games 
and a reduction of POMDPs to games with probabilistic uncertainty.
As a consequence we establish the precise decidability frontier for 
games with probabilistic uncertainty.
Table~\ref{tab:complexity} summarizes our results. 
For most problems we establish EXPTIME-complete bounds.
For some decidable problems we establish 2EXPTIME upper bounds, and EXPTIME lower
bounds, and establishing the precise complexity results 
are interesting open problems.

\end{document}